\documentclass[11pt]{scrartcl}
\usepackage{amsmath}
\usepackage{amssymb}
\usepackage{amsthm}
\usepackage{tikz}
\usepackage{pifont}
\usetikzlibrary{arrows}
\allowdisplaybreaks

\newtheorem{theorem}{Theorem}
\newtheorem{lemma}{Lemma}

\newtheorem{definition}{Definition}

\DeclareMathOperator*{\poly}{poly}

\newcommand{\lnln}{\ln\ln}

\title{Computing an Approximately Optimal Agreeable Set of Items\footnote{A preliminary version of this paper appeared in Proceedings of the 26th International Joint Conference on Artificial Intelligence, July 2017.}}
\author{
Pasin Manurangsi\\UC Berkeley
\and
Warut Suksompong\\Stanford University
}
\date{\vspace{-5ex}}

\begin{document}

\maketitle

\begin{abstract}
We study the problem of finding a small subset of items that is \emph{agreeable} to all agents, meaning that all agents value the subset at least as much as its complement. Previous work has shown worst-case bounds, over all instances with a given number of agents and items, on the number of items that may need to be included in such a subset. Our goal in this paper is to efficiently compute an agreeable subset whose size approximates the size of the smallest agreeable subset for a given instance. We consider three well-known models for representing the preferences of the agents: ordinal preferences on single items, the value oracle model, and additive utilities. In each of these models, we establish virtually tight bounds on the approximation ratio that can be obtained by algorithms running in polynomial time.
\end{abstract}

\section{Introduction}

A typical resource allocation problem involves dividing a set of resources among interested agents. We are often concerned with the \emph{efficiency} of the allocation, e.g., achieving high social welfare or ensuring that there is no other allocation that would make every agent better off than in the current allocation. Another important issue is the \emph{fairness} of the allocation. For example, we might want the resulting allocation to be \emph{envy-free}, meaning that every agent regards her bundle as the best among the bundles in the allocation \cite{Foley67,Varian74}, or \emph{proportional}, meaning that every agent obtains at least her proportionally fair share \cite{Steinhaus48}. A common feature of such problems is that one agent's gain is another agent's loss: the setting inherently puts the agents in conflict with one another, and our task is to try to resolve this conflict as best we can according to our objectives.

We consider a variant of the resource allocation problem where instead of the agents being pitted against one another, they belong to one and the same group. We will collectively allocate a subset of items to this group, and our goal is to make this subset ``agreeable'' to all agents. Agreeability can be thought of as a minimal desirability condition: While an agent may be able to find other subsets of items that she personally prefers, the current set is still acceptable for her and she can agree with the allocation of the set to the group. Without further constraints, the problem described so far would be trivial, since we could simply allocate the whole set of items to the agents. We therefore impose a constraint that the allocated subset should be small. This constraint on size is reasonable in a variety of settings. For instance, the agents could be going together on a trip and there is limited space in the luggage. Alternatively, they could be receiving some items in a resource allocation setting where the preferences of the other groups are not known or are given lower priority, perhaps because the groups have not arrived or are placed lower in a team competition, so we want the subset to be agreeable to the first group while leaving as many items as possible to the remaining groups. 

The problem of allocating a small agreeable subset of items was first studied by Suksompong, who defined the notion of agreeability based on the fairness notion of envy-freeness \cite{Suksompong16}. A subset of items is said to be \emph{agreeable} to an agent if the agent likes it at least as much as the complement set. Agreeability, or minor variants thereof, has been considered in the context of fair division, where each group consists of a single agent \cite{BouveretEnLa10,BramsKiKl12,AzizGaMa15}. In the example of agents going together on a trip, a subset of items that they take is agreeable if they weakly prefer it to the complement subset of items left at home. In the resource allocation example, assuming that we allocate resources between two groups, a subset is agreeable to the first group if no agent in the group would rather switch to the second group. Suksompong showed a tight upper bound on the number of items that may need to be included in the set in order for it to be agreeable to all agents. In particular, for every additional agent, the worst-case bound increases by approximately half an item. This result is rather surprising, since agents can have very different preferences on the items, and yet we only need to pay a relatively mild cost for each extra agent to keep agreeability satisfied for the whole group. When there are two or three agents, Suksompong also gave polynomial-time algorithms that compute an agreeable subset whose size matches the worst-case bound.

While Suksompong's results are quite intriguing, some important issues were left unaddressed by his work. Firstly, in many instances, an agreeable subset of smallest size is much smaller than the worst-case bound over all instances with that number of agents and items. Indeed, an extreme example is when there is a single item that every agent likes better than all of the remaining items combined. In this case, it suffices to allocate that item. This results in a much smaller set than the worst-case bound, which is at least half of the items for any number of agents. Secondly, even if we were content with finding a subset that matches the worst-case bound, we might not be able to compute it efficiently, thus rendering the existence result impractical when the number of agents or items is large. A related issue is that of eliciting the preferences on subsets of items from the agents. Since there are an exponential number of subsets, the burden on the agents to determine their preferences and the amount of information that they need to submit to our algorithm is potentially huge. This issue can be circumvented by relying on preferences over single items or allowing the algorithm to query the agents' preferences on a need-to-know basis, as is the case for Suksompong's polynomial time algorithm for two and three agents, respectively.\footnote{Note that if the algorithm elicits the whole preference relations from the agents, this elicitation step alone already prevents the algorithm from running in time polynomial in the number of items.}

In this paper, we address all of these issues and investigate the problem of computing an agreeable subset of approximately optimal size for a given instance, as opposed to one whose size is close to the worst-case bound over all instances with that number of agents and items. We tackle the problem using several models that are well-studied in the literature and present computationally efficient algorithms for computing an agreeable subset of approximately optimal size in each of them. Moreover, in all of the models we show that our approximation bounds are virtually tight.

In Section \ref{sec:ordinal}, we assume that we only have access to the agents' ordinal preferences on single items rather than subsets of items. Models of this type offer the advantage that the associated algorithms are often simple to implement and the agents do not need to give away or even determine their entire utility functions; such models have therefore received widespread attention \cite{KohlerCh71,BouveretEnLa10,AzizGaMa15}. With only the ordinal preferences on single items in hand, however, most of the time we cannot tell whether a certain subset is agreeable to an agent or not. Nevertheless, by assuming that preferences are responsive, we can extend preferences on single items to partial preferences on subsets. We show that for any constant number of agents, there exists a subset of $\frac{m}{2}+o(m)$ items that is always agreeable as long as the full responsive preferences are consistent with the rankings over single items. Since a necessarily agreeable subset always consists of at least $\frac{m}{2}$ items even for one agent, this bound is essentially tight. We also present a simple randomized algorithm and a deterministic algorithm, both running in polynomial time, to compute such a subset.

Next, in Section \ref{sec:general}, we consider general preferences using the value oracle model \cite{FeigeMiVo11}, where the preferences of the agents are represented by utility functions and we are allowed to query the utility of an agent for any subset. We exhibit an efficient approximation algorithm with approximation ratio $O(m\ln\ln m/\ln m)$ in this model. While this may not seem impressive since the trivial algorithm that always outputs the whole set of items already achieves approximation ratio $O(m)$, we also show that our ratio is essentially the best we can hope for. In particular, there does not exist a polynomial time algorithm with approximation ratio $o(m/\ln m)$. 

Finally, in Section \ref{sec:additive}, we assume that the agents are endowed with additive utility functions. Additivity provides a reasonable tradeoff between simplicity and expressiveness; it is commonly assumed in the literature, especially recently \cite{AmanatidisMaNi15,BouveretLe16,CaragiannisKuMo16}. We show that under additive valuations, it is NP-hard to decide whether there exists an agreeable set containing exactly half of the items even where there are only two agents. On the other hand, using results on covering integer programs, we demonstrate the existence of an $O(\ln n)$-approximation algorithm for computing a minimum size agreeable set. Moreover, we show that this is tight: It is NP-hard to approximate the problem to within a factor of $(1-\delta)\ln n$ for any $\delta>0$.

\section{Preliminaries}

Let $N=\{1,2,\ldots,n\}$ denote the set of agents, and $S=\{x_1,x_2,\dots,x_m\}$ the set of items. The agents in $N$ will be collectively allocated a subset of items in $S$. Denote by $\mathcal{S}$ the set of all subsets of $S$. Each agent $i$ is endowed with a preference relation $\succeq_i$, a reflexive, complete, and transitive ordering over $\mathcal{S}$. Let $\succ_i$ denote the strict part and $\sim_i$ the indifference part of the relation $\succeq_i$. For items $x$ and $y$, we will sometimes  write $x\succeq y$ to mean $\{x\}\succeq\{y\}$. We assume throughout the paper that preferences are \emph{monotonic}, i.e., $T\cup\{x\}\succeq T$ for all $T\subseteq S$ and $x\in S$.

We are interested in when a set of items is agreeable to an agent. To this end, we must precisely define what agreeability means. The notion of agreeability, defined by Suksompong \cite{Suksompong16}, is based on the fairness concept of envy-freeness. A subset is considered to be agreeable to an agent if the agent likes it at least as much as the complement set. Put differently, if the complement set is allocated to another agent, then the former agent does not envy the latter.

\begin{definition}
\label{def:agreeable}
A subset $T\subseteq S$ of items is said to be \emph{agreeable} to agent $i$ if $T\succeq_i S\backslash T$.
\end{definition}

Next, we define a property of preferences called \emph{responsiveness}, which says that an agent cannot be worse off whenever an item is added to her set or replaced by another item that she weakly prefers to the original item. Responsiveness is a reasonable assumption in many settings and has been widely studied in the literature \cite{BramsFi00,BramsKiKl12}.

\begin{definition}
\label{def:responsive}
A preference $\succeq$ on $\mathcal{S}$ is called \emph{responsive} if it satisfies the following two conditions:
\begin{itemize}
\item $T\cup\{x\}\succeq T$ for all $T\subseteq S$ and $x\in S$ (monotonicity);
\item $T\backslash \{y\}\cup \{x\}\succeq T$ for all $T\subseteq S$ and $x,y\in S$ such that $x\succeq y$, $x\not\in T$ and $y\in T$.
\end{itemize}
\end{definition}

When preferences are responsive, it sometimes suffices to know an agent's preference on single items (i.e., the restriction of $\succeq_i$ to subsets consisting of single items) in order to deduce that the agent regards a subset as agreeable. This intuition is formalized in the next definition.

\begin{definition}
\label{def:necessaryEF}
Fix a preference $\succeq^{sing}$ on single items in $S$. A subset $T\subseteq S$ is \emph{necessarily agreeable} with respect to $\succeq^{sing}$ if $T\succeq S\backslash T$ for any responsive preference $\succeq$ on $\mathcal{S}$ consistent with $\succeq^{sing}$.
\end{definition}

The following characterization of necessary agreeability will be useful for our results in Section \ref{sec:ordinal}.

\begin{lemma}[\cite{Suksompong16}]
\label{lem:necessaryEF}
Fix a preference $\succeq^{sing}$ on single items in $S$ with \[x_1\succeq^{sing}x_2\succeq^{sing}\cdots\succeq^{sing}x_m.\] 
Let $T\subseteq S$, and define $I_k=\{x_1,x_2,\ldots,x_k\}$ for all $k\in\{1,2,\ldots,m\}$. If \[|I_k\cap T|\geq\frac{k}{2}\] for all $k\in\{1,2,\ldots,m\}$, then $T$ is necessarily agreeable with respect to $\succeq^{sing}$. The converse also holds if the preference $\succeq^{sing}$ is strict.
\end{lemma}

\section{Ordinal Preferences on Single Items}
\label{sec:ordinal}

In this section, we assume that we only have at our disposal the agents' ordinal preferences on single items. We are interested in computing a small necessarily agreeable subset that is consistent with these preferences. If we had access to the agents' preferences over all subsets of items, it is known that we could always find a subset of size $\left\lfloor\frac{m+n}{2}\right\rfloor$ that is agreeable to all $n$ agents \cite{Suksompong16}. It is not clear, however, how much extra ``penalty'' we have to pay for the information restriction that we are imposing. It could be, for example, that with three agents there exist preferences for which we have to include up to $\frac{2m}{3}$ or $\frac{3m}{4}$ items in a necessarily agreeable subset. We show that this is in fact not the case---there always exists a necessarily agreeable subset of size $\frac{m}{2}+O(\sqrt{m\log\log m})$ as long as the number of agents is constant. 

For the proof of our first result, we will require the law of the iterated logarithm, which gives a bound for the fluctuations of a random walk.

\begin{lemma}[Law of the iterated logarithm~\cite{Khintchine24,Kolmogoroff29}]
\label{lem:iteratedlog}
Let $X_1,X_2,\ldots$ be independent and identically distributed random variables with mean $0$ and variance $1$. Let $S_n:=X_1+\dots+X_n$. Then
\[
\limsup_{n\rightarrow\infty}\frac{S_n}{\sqrt{n\log\log n}}=\sqrt{2}
\]
almost surely.
\end{lemma}

\begin{theorem}
\label{thm:necessaryrandomized}
If the number of agents is constant, then there exists a subset of size $\frac{m}{2}+O(\sqrt{m\log\log m})$ that is necessarily agreeable with respect to the preferences on single items of all agents.
\end{theorem}

We remark that when there are two agents, Suksompong's algorithm computes a necessarily agreeable subset of size $\left\lfloor\frac{m}{2}\right\rfloor+1$, which is also the optimal bound for an agreeable subset for two agents even if we know their full preferences.

\begin{proof}
Let $X_1,X_2,\ldots,X_m$ be independent random variables taking values $1$ or $-1$. We will take $X_i=1$ to mean that item $x_i$ is included in our subset, and $X_i=-1$ to mean that it is excluded from the subset. 

Consider an arbitrary agent $j$. Suppose that she ranks the single items as \[x_{\sigma_j(1)}\succeq_j x_{\sigma_j(2)}\succeq_j\dots\succeq_j x_{\sigma_j(m)}.\] Using Lemma \ref{lem:necessaryEF}, we find that our subset is necessarily agreeable for the agent if \[X_{\sigma_j(1)}+\dots+X_{\sigma_j(i)}\geq 0\] for all $1\leq i\leq m$.

Let $S_i^j:=X_{\sigma_j(1)}+\dots+X_{\sigma_j(i)}$, and independently set each $X_i$ to be $1$ or $-1$ with probability $\frac{1}{2}$ each. Lemma \ref{lem:iteratedlog} implies that for large enough $k$ (and $m$), we have $S_i^j\leq 2\sqrt{i\log\log i}$ (and by symmetry, $-S_i^j\geq-2\sqrt{i\log\log i}$) for all $i\geq k$ with high probability. This means that for large enough $m$, $|S_i^j|\leq 2\sqrt{m\log\log m}$ for all $1\leq i\leq m$ with high probability as well. 

Since $n$ is constant, by the union bound we find that for sufficiently large $m$, we can initialize the random variables $X_1,\dots,X_m$ so that $|S_i^j|\leq2\sqrt{m\log\log m}$ for all $1\leq i\leq m$ and all $j$. We will modify the choice of $X_i$'s slightly. For each agent, include her most preferred $\sqrt{m\log\log m}$ items that have so far been excluded. Thus, we have $S_i^j\geq 0$ for all $1\leq i\leq m$ and all $j$, and our selected subset includes at most 
\[\frac{m}{2}+(n+1)\cdot\sqrt{m\log\log m}=\frac{m}{2}+O(\sqrt{m\log\log m})\] 
items, as desired.
\end{proof}

Theorem \ref{thm:necessaryrandomized} yields a simple  randomized polynomial time algorithm that finds a necessarily agreeable subset for multiple agents by first choosing independently and uniformly at random whether to include each item, and then modifying the selection by including the most preferred $\sqrt{m\log\log m}$ items of each agent that have been excluded.

Next, we present a deterministic polynomial time algorithm that finds a necessarily agreeable subset of size $m/2+o(m)$ for all agents. We will need the following classical result in combinatorics.

\begin{lemma}[\cite{ErdosSz35}]
\label{lem:longestsubsequence}
Any sequence of at least $(r-1)(s-1)+1$ distinct real numbers contains either an increasing subsequence of length $r$ or a decreasing subsequence of length $s$.
\end{lemma}

\begin{theorem}
\label{thm:necessarydeterministic}
If the number of agents is constant, then there exists a deterministic algorithm, running in time polynomial in the number of items, that computes a subset of size $\frac{m}{2}+o(m)$ that is necessarily agreeable with respect to the preferences on single items of all agents.
\end{theorem}

\begin{proof}
We first observe that when there are $m$ items and two agents whose preferences on these items are the opposite of each other, we can choose a subset of size at most $\frac{m}{2}+1$ that is necessarily agreeable to both agents. To see this, suppose that the preferences on single items of the two agents are \[x_1\succeq_1 x_2\succeq_1\dots\succeq_1 x_m\] and \[x_m\succeq_2 x_{m-1}\succeq_2\dots\succeq_2 x_1.\] If $m$ is odd, it suffices to choose items $x_1,x_3,x_5,\dots,x_m$, while if $m$ is even, it suffices to choose items $x_1,x_3,x_5,\dots,x_{m-1}$ along with item $x_m$.

Now, consider our group of $n$ agents. Assume without loss of generality that the preference on single items of agent 1 is \[x_1\succeq_1 x_2\succeq_1\dots\succeq_1 x_m.\] Using Lemma \ref{lem:longestsubsequence}, from the preference on single items of agent 2, we can find either an increasing subsequence or a decreasing subsequence of length $\sqrt{m}$ on the indices of the items. Applying Lemma \ref{lem:longestsubsequence} again, we find a subsequence of this subsequence of length $m^{1/4}$ that is either increasing or decreasing in the preference on single items of agent 3. Proceeding in this manner, we find a subsequence of $m^{1/2^{n-1}}$ items whose indices appear either in that order or in the reverse order in every agent's preference on single items. By our observation above, we can choose a subset of at most $\frac{1}{2} m^{1/2^{n-1}}+1$ items that is necessarily agreeable for every agent with respect to this set of $m^{1/2^{n-1}}$ items.

Let $t:=2^{n-1}$. Note that if a set $A$ of items is necessarily agreeable when the universe of items is taken to be $B\supseteq A$, and another set $C$ of items is necessarily agreeable with respect to the same preference on single items when the universe of items is $D\supseteq C$ disjoint from $B$, then $A\cup C$ is necessarily agreeable with respect to that preference when the universe of items is $B\cup D$. To obtain our necessarily agreeable subset for all agents, we proceed as follows. When there are $k$ items left, we choose a subset of $k^{\frac{1}{t}}$ items as above. Within that subset, we choose a subset of at most $\frac{1}{2} k^{\frac{1}{t}}+1$ items that is necessarily agreeable for all agents with respect to the $k^{\frac{1}{t}}$ items, and we remove the $k^{\frac{1}{t}}$ items from consideration. In the first step, we decrease the number of items by a factor of $1-1/m^{(t-1)/t}$, and this factor only decreases in subsequent steps. Hence the number of steps is at most
\[\frac{\log\frac{1}{m}}{\log\left(1-\frac{1}{m^{(t-1)/t}}\right)}\approx\frac{\log m}{\frac{1}{m^{(t-1)/t}}}=m^{\frac{t-1}{t}}\log m,\]
which implies that our chosen subset exceeds half of the items by $o(m)$ items. Moreover, since each step involves finding a longest increasing or decreasing subsequence from a list of length at most $m$ for a constant number of agents, the algorithm runs in polynomial time.
\end{proof}

Although the algorithm in Theorem~\ref{thm:necessarydeterministic} has the advantage of being deterministic, the expected number of repetitions of the randomized algorithm in Theorem~\ref{thm:necessaryrandomized} is very low; in fact, this value is roughly 1 since the algorithm succeeds with high probability. Therefore, we think that the algorithm from Theorem 1 should be preferred in general due to its speed and ease of implementation as well as its superior guarantee.

\section{General Preferences}
\label{sec:general}

While our algorithms from the previous section always find an agreeable set of size at most $\frac{m}{2} + o(m)$, it is unclear how small this set is compared to the optimal if we have information beyond preferences on single items. In other words, our results so far do not yield any guarantee on the approximation ratio beyond the obvious $O(m)$ upper bound for arbitrary preferences over subsets of items. The goal of this section is to explore the possibilities and limitations of achieving better approximation ratios in this general setting.

Before we move on to our results, let us be more precise about the model we are working with. First, since each agent's preference is reflexive, complete and transitive, there is a utility function $u_i: \mathcal{S} \rightarrow [0, 1]$ such that $T \succeq_i T'$ if and only if $u_i(T) \geq u_i(T')$; for convenience, we work with these utility functions instead of working directly with the preferences themselves. Since the number of subsets in $\mathcal{S}$ is exponentially large, the utility functions take exponential space to write down. Hence, it is undesirable to include them as part of the input. Instead, we will work with the \emph{value oracle model}~\cite{FeigeMiVo11}, in which the algorithm can query $u_1(T), \dots, u_n(T)$ for each subset $T \subseteq S$.\footnote{When a polynomial number of queries are allowed, it is not hard to see that the value oracle model can be simulated with the preference oracle model (e.g., \cite{Suksompong16}), in which the algorithm is allowed to query the relative preference of an agent between any two subsets, and vice-versa.} Finally, we note that we do not assume responsiveness of the agents' preferences in this section.

Our first result is a simple polynomial time approximation algorithm with approximation ratio $O(m \lnln m / \ln m)$. Even though this approximation guarantee is only $\Omega(\ln m / \lnln m)$ better than the obvious $O(m)$ bound, we will see later that this is almost the best one can hope for in polynomial time.

\begin{theorem} \label{thm:approx-general}
For any constant $\varepsilon > 0$, there exists a deterministic $(\varepsilon m \lnln m / \ln m)$-approximation algorithm for finding a minimum size agreeable set that runs in time polynomial in the number of agents and items.
\end{theorem}

Our algorithm works as follows. First, we query the agents' utilities on $(\ln m / \lnln m)$-size subsets $T_1, \dots, T_{\poly(m)}$, which are to be specified. If any of these subsets are agreeable, then we output one such subset. Otherwise, we output the whole set $S$.

To argue about the approximation ratio of the algorithm, we need to show that, if there is a small agreeable subset $T^* \subseteq S$ (of size $O(\ln m / \lnln m)$), then at least one of $T_1, \dots, T_{\poly(m)}$ is agreeable. A sufficient condition for this is that every subset $T \subseteq S$ of small size is contained in at least one $T_i$. A similar question has been studied before in combinatorics under the name \emph{covering design}~\cite{GordonPaKu96}. Here we will use a construction by Rees et al., stated formally below.

\begin{lemma}[\cite{ReesStWe99}] \label{lem:covering-design}
For any set $S$ with $|S|=m$ and any positive integers $p, q$ such that $pq \leq m$, there exists a deterministic algorithm that outputs subsets $T_1, \dots, T_{\binom{\lceil m / p \rceil}{q}} \subseteq S$ of size $pq$ such that, for any subset $T \subseteq S$ of size at most $q$, $T$ is contained in $T_i$ for some $i$. Moreover, the running time of the algorithm is polynomial in $m$ and $\binom{\lceil m / p \rceil}{q}$.
\end{lemma}


We give a short proof of Lemma~\ref{lem:covering-design}, which is taken from Rees et al.'s work but modified with the desired range of the parameters.

\begin{proof}[Proof of Lemma~\ref{lem:covering-design}]
First, partition $S$ into $\lceil m/p \rceil$ parts $S_1, \dots, S_{\lceil m / p \rceil}$ where each part is of size at most $p$. Each set $T_i$ is simply a union of $q$ different $S_j$'s. Note that if the union is of size smaller than $pq$, we can simply add arbitrary elements of $S$ into it to make its size exactly $pq$. Clearly, there are $\binom{\lceil m / p \rceil}{q}$ such unions and $T_1, \dots, T_{\binom{\lceil m / p \rceil}{q}}$ satisfy the properties required in the theorem.
\end{proof}

With this lemma in place, we are ready to prove Theorem~\ref{thm:approx-general}.

\begin{proof}[Proof of Theorem~\ref{thm:approx-general}]
Our algorithm starts by evoking the algorithm from Lemma~\ref{lem:covering-design} with $q = \lfloor \ln m / (\varepsilon \lnln m) \rfloor$ and $p = \lfloor \varepsilon m \lnln m / (q\ln m) \rfloor$ to produce subsets $T_1, \dots, T_\ell \subseteq S$ of size $pq$ where $\ell = \binom{\lceil m / p \rceil}{q}$. For each $i = 1, \dots, \ell$, we query $u_1(T_i), \dots, u_n(T_i), u_1(S \setminus T_i), \dots, u_n(S \setminus T_i)$ to check whether $T_i$ is agreeable. If any $T_i$ is agreeable, we output it. Otherwise, output the whole set $S$. Clearly, the output set is always agreeable.

Moreover, observe that, from our choice of $p, q$, we have 
\begin{align*}
\ell = \binom{\lceil m / p \rceil}{q} 
     &\leq \left(e \lceil m / p \rceil / q\right)^q \\
     &= (O(\ln m))^q \\
     &= \exp(O(q \lnln m)) \\
     &= \exp(O(\ln m)), 
\end{align*}
which is polynomial in $m$. Hence, the running time of our algorithm is polynomial in $m$ and $n$.

To prove the algorithm's approximation guarantee, let us consider two cases. First, if the optimal agreeable set has size more than $q$, then since we output a set of size at most $m$, we obtain an approximation ratio of at most $m / (q + 1) \leq \varepsilon m \lnln m / \ln m$.

On the other hand, if the optimal agreeable set $T^*$ has size at most $q$, then by construction, there exists $i$ such that $T^* \subseteq T_i$. Since $T^*$ is agreeable, $T_i$ is also agreeable. Hence, our output set has size $pq \leq \varepsilon m \lnln m / \ln m$, which implies an approximation ratio of $\varepsilon m \lnln m / \ln m$ as well.
\end{proof}

While our algorithm may seem rather naive, we will show next that its approximation guarantee is, up to $O(\lnln m)$ factor, essentially the best one can hope for, even when there is only one agent:

\begin{theorem} \label{thm:inapprox-general}
For every constant $c> 0$ and every sufficiently large $m$ (depending on $c$), there is no (possibly randomized and adaptive) algorithm that makes at most $m^{c/8}$ queries and always outputs an agreeable set with expected size at most $m / (c \ln m)$ times the optimum, even when there is only one agent.
\end{theorem}

In other words, the above theorem implies that there is no polynomial time algorithm with approximation ratio $o(m / \ln m)$. We note here that our lower bound is information-theoretic and is not based on any computational complexity assumptions. Moreover, it rules out any algorithm that makes a polynomial number of queries, not only those that run in polynomial time.

\begin{proof}[Proof of Theorem~\ref{thm:inapprox-general}]
Let $g: \mathcal{S} \rightarrow [0, 1]$ be defined by
\begin{align*}
g(T) =
\begin{cases}
1 & \text{ if } |T| \geq m / 2, \\
0 & \text{ otherwise.}
\end{cases}
\end{align*}
Moreover, for each subset $T^* \subseteq S$, let $f_{T^*}: \mathcal{S} \rightarrow [0, 1]$ denote the function
\begin{align*}
f_{T^*}(T) =
\begin{cases}
1 & \text{ if } |T| \geq m / 2 \text{ or } T^* \subseteq T, \\
0 & \text{ otherwise.}
\end{cases}
\end{align*}

That is, $f_{T^*}$ is $g$ together with a planted solution $T^*$.

Consider any algorithm $\mathcal{A}$ that makes at most $m^{c/8}$ queries. Assume for the moment that $\mathcal{A}$ is deterministic. Let us examine a run of $\mathcal{A}$ when the agent's utility is $g$. Suppose that $\mathcal{A}$'s queries to $g$ are $T_1, \dots, T_{\lfloor m^{c/8}\rfloor} \subseteq S$.

Let $T^*$ be a random subset of $S$ of size $\lfloor c \ln m / 4 \rfloor$. Let us now consider the queries $\mathcal{A}$ made when the agent's utility is $f_{T^*}$; suppose that the queries made are $T'_1, \dots, T'_{\lfloor m^{c/8}\rfloor} \subseteq S$. For every $j = 1, \dots, \lfloor m^{c/8}\rfloor$, if $T_1 = T'_1, \dots, T_{j - 1} = T'_{j - 1}$ and $g(T_1) = f_{T^*}(T'_1), \dots, g(T_{j - 1}) = f_{T^*}(T'_{j - 1})$, then $\mathcal{A}$ goes through the same computation route for both $g$ and $f_{T^*}$, and hence $T_j = T'_j$. Moreover, when both runs share the same computational route so far and $T_j = T'_j$, we can bound the probability that $g(T_j) \ne f_{T^*}(T'_j)$ as follows. First, if $|T_j| \geq m / 2$, then $g(T_j)$ is always equal to $f_{T^*}(T'_j)$. Otherwise, since $T_j$ is independent of $T^*$, we have 
\begin{align*}
\Pr[g(T_j) \neq f_{T^*}(T'_j)] &= \Pr[g(T_j) \neq f_{T^*}(T_j)] \\
&= \Pr[T^* \subseteq T_j] \\
&= \frac{\binom{|T_j|}{\lfloor c \ln m / 4 \rfloor}}{\binom{n}{\lfloor c \ln m / 4 \rfloor}} \\
&\leq \left(\frac{|T_j|}{n}\right)^{\lfloor c \ln m / 4 \rfloor} \\
&\leq 2^{-\lfloor c \ln m / 4 \rfloor} \\
&\leq 2 m^{-c / 6}.
\end{align*}

Hence, by union bound, the probability that the two sequences of queries are not identical is at most $(2 m^{-c / 6}) \cdot m^{c / 8} = 2 m^{- c / 24}$, which is less than $1/2$ when $m$ is sufficiently large. Furthermore, observe that when the two sequences are identical, $\mathcal{A}$ must output an agreeable subset for $g$, which is of size at least $m / 2$. Thus, the expected size of the output of $\mathcal{A}$ on $f_{T^*}$ is more than $m / 2 \cdot (1 / 2) = m / 4$. However, the optimal agreeable set for $f_{T^*}$ has size only $\lfloor c \ln m / 4 \rfloor$. As a result, the expected size of the output of $\mathcal{A}$ is more than $m / (c \ln m)$ times the optimum as desired.

Finally, note that if $\mathcal{A}$ is randomized, we can use the above argument on each choice of randomness and average over all the choices, which gives a similar conclusion.
\end{proof}

\section{Additive Utilities}
\label{sec:additive}

In this section, we assume that the agents' preferences are represented by additive utility functions. Each agent $i$ has some nonnegative utility $u_i(x_j)$ for item $x_j$, and $u_i(T)=\sum_{x\in T}u_i(x)$ for any $i\in N$ and any subset of items $T\subseteq S$. 

Clearly, the problem of deciding whether there exists an agreeable set of a certain size is in NP. Moreover, the following theorem shows that it is indeed NP-complete, even when there are two agents with additive utility functions.

\begin{theorem}
\label{thm:twoplayershardness}
Even when there are two agents with additive utility functions, it is NP-hard to decide whether there is an agreeable set of size exactly $\frac{m}{2}$.
\end{theorem}

\begin{proof}
We will reduce from the following problem called {\normalfont \scshape Balanced 2-Partition}: given a multiset $A$ of non-negative integers, decide whether there exists a subset $B \subseteq A$ such that $|B| = |A \setminus B| = |A|/2$ and $\sum_{a \in B} a = \sum_{a \in A \setminus B} a = \sum_{a \in A} a / 2$. 

Like the well-known {\normalfont \scshape 2-Partition} where the cardinality constraint is not included, {\normalfont \scshape Balanced 2-Partition} is NP-hard. 
For completeness, we give a simple proof of NP-hardness of {\normalfont \scshape Balanced 2-Partition} in Appendix~\ref{app:np-hard-part}.

The reduction from {\normalfont \scshape Balanced 2-Partition} proceeds as follows. Let $a_1, \dots, a_{|A|}$ be the elements of $A$. The set $S$ contains $|A|$ items $x_1, \dots, x_{|A|}$, each associated with an element of $A$. The utility function is then defined by $u_1(x_i) = a_i$ and $u_2(x_i) = M - a_i$ where $M = \sum_{a \in A} a$. We will next show that this reduction is indeed a valid reduction.

(YES Case) Suppose that there exists $B \subseteq A$ such that $|B| = |A|/2$ and $\sum_{a \in B} a = \sum_{a \in A} a / 2$. We can simply pick $T$ to be all the items corresponding to the elements in $B$. It is obvious that $T$ has size $|A| / 2 = m / 2$ and that $T$ is agreeable.

(NO Case) We prove the contrapositive; suppose that there is an agreeable subset $T \subseteq S$ of size $m/2$. Let $B$ be the corresponding elements in $A$ of all the items in $T$. Since $T$ is agreeable, $\sum_{x \in T} u_i(x) \geq \sum_{x \in S \setminus T} u_i(x)$ for $i \in \{1, 2\}$. When $i = 1$, this implies that $\sum_{a \in B} a \geq \sum_{a \in A} a / 2$. Moreover, when $i = 2$, using the fact that $|T| = m / 2$, we have $\sum_{a \in B} a \leq \sum_{a \in A} a / 2$. Thus, $\sum_{a \in B} a = \sum_{a \in A} a / 2$. Finally, note that $|B| = m / 2 = |A| / 2$. Hence, $A$ is a YES instance for {\normalfont \scshape Balanced 2-Partition}.
\end{proof}

Theorem~\ref{thm:twoplayershardness} shows that the problem is weakly NP-hard even when there are two agents. Nevertheless, when the number of agents is constant, there exists a pseudo-polynomial time dynamic programming algorithm for computing an optimal agreeable set. In particular, the problem is not strongly NP-hard for a constant number of agents.

\begin{theorem}
If the number of agents is constant, then there exists a pseudo-polynomial time algorithm that computes an agreeable set of minimum size.
\end{theorem}

\begin{proof}
The algorithm uses dynamic programming. Assume that the utilities of agent $i$ for the items are integers with sum $\sigma_i$. We construct a table $\Sigma$ of size $(m+1)(\sigma_1+1)\dots(\sigma_n+1)$, where for each $0\leq m'\leq m$ and each tuple $(y_1,\dots,y_n)$ with $0\leq y_i\leq \sigma_i$, the entry $\Sigma(m',y_1,\dots,y_n)$ of the table corresponds to the minimum number of items from among the first $m'$ items that we need to include so that agent $i$ has utility exactly $y_i$ for all $i$ (if this is achievable). Initially we have $\Sigma(0,0,\dots,0)=0$ and $\Sigma(m',y_1,\dots,y_n)=\infty$ otherwise. We then iterate through the values of $m'$ in increasing order. For each $m'\geq 1$, we update the entries of the table as follows:
\begin{itemize}
\item If $u_i(x_{m'})\leq y_i$ for all $i$ and $$1+\Sigma\left(m'-1,y_1-u_1(x_{m'})\dots,y_n-u_n(x_{m'})\right)<\Sigma(m'-1,y_1,\dots,y_n),$$ let $\Sigma(m',y_1,\dots,y_n)=1+\Sigma\left(m'-1,y_1-u_1(x_{m'})\dots,y_n-u_n(x_{m'})\right)$. 
\item Else, let $\Sigma(m',y_1,\dots,y_n)=\Sigma(m'-1,y_1,\dots,y_n)$.
\end{itemize}

Finally, we look up the entries $\Sigma(m,y_1,\dots,y_n)$ for which $y_i\geq \frac{\sigma_i}{2}$ holds for all $i$ and return the minimum value over all such entries. This algorithm takes time $O(m\sigma_1\dots \sigma_n)$. Note that if we also want to return an agreeable set (rather than just the minimum size), we can also keep track of the sets of items along with the values in our table. 
\end{proof}

While there is a pseudo-polynomial time algorithm for the problem when the number of agents is fixed, we show below that when the number of agents is not fixed, the problem becomes strongly NP-hard. In other words, there is no pseudo-polynomial time algorithm for this variation unless P=NP. 

\begin{theorem}
When the number of agents is given as part of the input, it is strongly NP-hard to decide whether there is an agreeable set of size exactly $\frac{m+1}{2}$.
\end{theorem}

\begin{proof}
We will reduce from {\normalfont \scshape 3SAT}. Given a 3SAT formula $\phi$ with $m'$ clauses $C_1, \dots, C_{m'}$ on $n'$ variables $y_1, \dots, y_{n'}$, let there be $n=m' + n'$ agents $C_1, \dots, C_{m'}, y_1, \dots, y_{n'}$ and $m=2n' + 1$ items where $2n'$ items correspond to all the literals $y_1, \neg y_1, \dots, y_{n'}, \neg y_{n'}$ and the remaining item is called $a$. The utility function is defined by
\begin{align*}
u_{C_i}(b) =
\begin{cases}
1 & \text{ if } b = a, \\
1 & \text{ if the literal } b \text{ is present in } C_i, \\
0 & \text{ otherwise,}
\end{cases}
\end{align*}
and
\begin{align*}
u_{y_i}(b) =
\begin{cases}
1 & \text{ if } b = a, \\
1 & \text{ if } b = y_i \text{ or } b=\neg y_i \\
0 & \text{ otherwise.}
\end{cases}
\end{align*}
We will next show that this is a valid reduction. First, note that all the integer parameters are now polynomial in the size of the input. Hence, we are left to show that YES and NO instances of 3SAT map to YES and NO instances of our problem respectively.

(YES Case) Suppose that there exists an assignment that satisfies $\phi$. For each $y_i$, let $b_i$ be the literal of $y_i$ that is true according to this assignment. Let the set $T$ be $\{a, b_1, \dots, b_{n'}\}$. Since each $C_j$ is satisfied by the assignment, we have $\sum_{i=1}^{n'} u_{C_j}(b_i) \geq 1$. Hence, we have $\sum_{x \in T} u_{C_j}(x) \geq 2$, which implies that $T \succeq_{C_j} S \setminus T$. 

Moreover, for each variable $y_i$, $\sum_{x \in T} u_{y_i}(x) = 2$, which also implies that $T \succeq_{y_i} S \setminus T$. As a result, $T$ is an agreeable set of size $n' + 1 = \frac{m+1}{2}$ as desired.

(NO Case) We again prove the contrapositive; suppose that there exists an agreeable set $T \subseteq S$ of size $\frac{m+1}{2} = n' + 1$. We can assume without loss of generality that $a \in T$. Indeed, since the utility of any agent for $a$ is at least as much as the utility of the agent for any other item, if $a \notin T$, we can remove an arbitrary item from $T$ and add $a$ in while maintaining the agreeability of $T$. Moreover, we also assume without loss of generality that each clause of $\phi$ has at least two variables---it is obvious that every 3SAT formula can be transformed to this form in polynomial time.

Since $T \succeq_{y_i} S \setminus T$, at least one literal corresponding to $y_i$ is included in $T$. Moreover, since the size of $T$ is $n' + 1$ and $a \in T$, exactly one literal of each $y_i$ is in $T$; let $b_i$ be this literal. Consider the assignment to the variables such that all the $b_i$'s are satisfied. Since $T \succeq_{C_j} S \setminus T$ for every $C_j$ and $C_j$ contains at least two literals, at least one literal in $C_j$ is satisfied by this assignment. Hence, this assignment satisfies the formula $\phi$.
\end{proof}

Given that finding an agreeable set of minimum size is NP-hard, it is natural to attempt to find an approximation algorithm for the problem. When the utilities are additive, this turns out to be closely related to approximating the classical problem {\normalfont \scshape Set Cover}. In {\normalfont \scshape Set Cover}, we are given a ground set $U$ and a collection $\mathcal{C}$ of subsets of $U$. The goal is to select a minimum number of subsets whose union is the entire set $U$. 

{\normalfont \scshape Set Cover} was one of the first problems shown to be NP-hard in Karp's seminal paper~\cite{Karp72}. Since then, its approximability has been intensely studied and has, by now, been well understood. A simple greedy algorithm is known to yield a $(\ln |U| + 1)$-approximation for the problem~\cite{Johnson74,Lovasz75}. On the other hand, a long line of work in hardness of approximation~\cite{LundYa94,RazSa97,Feige98,AlonMoSa06,Moshkovitz15} culminates in Dinur and Steurer's work, in which a $(1 - \varepsilon)\ln |U|$ ratio NP-hardness of approximation for {\normalfont \scshape Set Cover} was proved for every constant $\varepsilon > 0$ \cite{DinurSt14}.

The first connection we will make between {\normalfont \scshape Set Cover} and approximating minimum size agreeable set is on the negative side---we will show that any inapproximability result of {\normalfont \scshape Set Cover} can be translated to that of approximating minimum size agreeable set as well. To do so, we will first state Dinur and Steurer's result more precisely.

\begin{lemma}[\cite{DinurSt14}]
For every constant $\varepsilon > 0$, there is a polynomial time reduction from any {\normalfont \scshape 3SAT} formula $\phi$ to a {\normalfont \scshape Set Cover} instance $(U, \mathcal{C})$ and $f(U) = \poly(|U|)$ such that
\begin{itemize}
\item (Completeness) if $\phi$ is satisfiable, the optimum of $(U, \mathcal{C})$ is at most $f(U)$.
\item (Soundness) if $\phi$ is unsatisfiable, the optimum of $(U, \mathcal{C})$ is at least $((1 - \varepsilon) \ln |U|) f(U)$.
\end{itemize}
\end{lemma}

We are now ready to prove hardness of approximation for minimum size agreeable set.

\begin{theorem}
For any constant $\delta>0$, it is NP-hard to approximate minimum size agreeable set to within a factor $(1 - \delta) \ln n$ of the optimum.
\end{theorem}

\begin{proof}
Let $\varepsilon = \delta / 2$. Given a 3SAT formula $\phi$, we first use Dinur-Steurer reduction to produce a {\normalfont \scshape Set Cover} instance $(U, \mathcal{C})$. Let there be $|U|$ agents, each of whom is associated with each element of $U$; it is convenient to think of the set of agents as simply $N = U$. As for the items, let there be one item for each subset $C \in \mathcal{C}$ and additionally let there be one special item called $t$; in other words, $S = \mathcal{C} \cup \{t\}$.

The utility for each $a \in U$ is then defined by
\begin{align*}
u_a(s) =
\begin{cases}
|\{C \in \mathcal{C} \mid a \in C\}| - 1 & \text{ if } s = t, \\
1 & \text{ if } s \in \mathcal{C} \text{ and } a \in s, \\
0 & \text{ otherwise}.
\end{cases}
\end{align*}
We show next that this reduction indeed gives the desired inapproximability result.

(Completeness) If $\phi$ is satisfiable, then there are $f(|U|)$ subsets from $\mathcal{C}$ that together cover $U$. We can take $T$ to contain all these subsets and the special item $t$. Clearly, $T$ has size $f(|U|) + 1$ and is agreeable.

(Soundness) If $\phi$ is unsatisfiable, then any set cover of $(S, \mathcal{C})$ contains at least $((1 - \varepsilon) \ln |U|)f(|U|)$ subsets. Consider any agreeable set $T$. For each $a \in U$, from our choice of $u_a(t)$, $T$ must include at least one subset that contains $a$. In other words, $T \setminus \{t\}$ is a set cover of $(S, \mathcal{C})$. Hence, $|T|$ must also be at least $((1 - \varepsilon) \ln |U|)f(|U|)$. 

Thus, it is NP-hard to approximate minimum size agreeable set to within $\frac{((1 - \varepsilon) \ln |U|)f(|U|)}{f(|U|) + 1}$ of the optimum. This ratio is at least $(1 - \delta) \ln n$ when the number of agents, $n = |U|$, is sufficiently large.
\end{proof}

Unlike the above inapproximability result, it is unclear how algorithms for {\normalfont \scshape Set Cover} can be used to approximate minimum size agreeable set.  Fortunately, our problem is in fact a special case of a generalization of {\normalfont \scshape Set Cover} called {\normalfont \scshape Covering Integer Program} ({\normalfont \scshape CIP}), which can be written as follows:
\begin{align*}
\text{minimize } & c^T x \\
\text{subject to } & Ax \geq 1, \\
& 0 \leq x \leq u \\
& x \in \mathbb{Z}^m
\end{align*}
where $c, u \in \mathbb{R}^m$ and $A \in \mathbb{R}^{n \times m}$ are given as input.

The problem of finding a minimum size agreeable set can then be formulated in this form simply by setting $c, u$ and $A$ as follows.
\begin{align*}
c_s = 1 & & \forall s \in S \\
u_s = 1 & & \forall s \in S \\
A_{i, s} = \frac{2 u_i(s)}{\sum_{s' \in S} u_i(s')} & & \forall i \in N, s \in S
\end{align*}

Similarly to {\normalfont \scshape Set Cover}, approximability of {\normalfont \scshape CIP} has been well studied; specifically, the problem is known to be approximable to within $O(\ln n)$ of the optimum in polynomial time~\cite{KolliopoulosYo05}. This immediately implies an $O(\ln n)$-approximation algorithm for finding a minimum size agreeable set as well:
\begin{theorem}
When the agents' utility functions are additive, there is a polynomial time $O(\ln n)$-approximation algorithm for finding a minimum size agreeable set.
\end{theorem}

\section{Conclusion}

In this paper, we study the problem of efficiently computing an approximately optimal agreeable set of items. We consider three well-known models for representing the preferences of the agents and derive essentially tight bounds on the approximation ratio that can be achieved in polynomial time in each model. When we only have access to agents' rankings on single items and the number of agents is constant, we can efficiently compute a necessarily agreeable subset of size $\frac{m}{2}+O(\sqrt{m\log\log m})$; this is almost tight since any such subset must contain at least $\frac{m}{2}$ items. Next, for general preferences, we exhibit an efficient algorithm with approximation ratio $O(m\ln\ln m/\ln m)$ under the value oracle model, and we show that no efficient algorithm can achieve approximation ratio $o(m/\ln m)$. On the other hand, for additive valuations, an algorithm with approximation ratio $O(\ln n)$ exists for computing a minimum size agreeable set; this is complemented by the result that it is NP-hard to approximate the problem to within a factor of $(1-\delta)\ln n$ for any $\delta>0$.

We conclude the paper by listing some directions for future work. Firstly, in the ranking model, it would be interesting to close the gap between our upper bound of $\frac{m}{2}+O(\sqrt{m\log\log m})$ and the lower bound of $\frac{m}{2}+O(1)$ by Suksompong \cite{Suksompong16} on the minimum size necessarily agreeable set. Secondly, one can consider a probabilistic setting, where the utilities of the agents are drawn from some distributions, and determine the size of the smallest agreeable set that is likely to exist in such a setting. Thirdly, the question of truthfulness has not been addressed in this or previous work. One can ask whether similar approximation ratios can be obtained or whether additional limitations arise in the presence of strategic agents. Finally, a related line of research, which has recently received attention, is to design algorithms that fairly allocate items among groups of agents \cite{ManurangsiSu17,Suksompong17}.

\bibliographystyle{abbrv}
\bibliography{main-full}

\appendix

\section{NP-hardness of {\normalfont \scshape Balanced 2-Partition}} \label{app:np-hard-part}

In this section, we show that {\normalfont \scshape Balanced 2-Partition} is NP-hard via a reduction from {\normalfont \scshape 2-Partition}, a well-known NP-hard problem.

\begin{lemma}
{\normalfont \scshape Balanced 2-Partition} is NP-hard.
\end{lemma}

\begin{proof}
We reduce from {\normalfont \scshape 2-Partition}, a problem in which a set $B$ of positive integers is given and the goal is to decide whether there exists a set $T \subseteq B$ such that $\sum_{b \in T} b = \sum_{b \in B \setminus T} b$. {\normalfont \scshape 2-Partition} is known to be NP-complete (see, e.g.,~\cite{GareyJo79}).

Given a {\normalfont \scshape 2-Partition} instance, we create a {\normalfont \scshape Balanced 2-Partition} instance as follows. Let $A$ be the multiset containing all elements of $B$ and $|B|$ additional zeros. Clearly, the reduction runs in polynomial time. We will next show that $B$ is a YES instance of {\normalfont \scshape 2-Partition} if and only if $A$ is a YES instance of {\normalfont \scshape Balanced 2-Partition}.

(YES Case) Suppose that $B$ is a YES instance of {\normalfont \scshape 2-Partition}, i.e., there exists $T \subseteq B$ such that $\sum_{b \in T} b = \sum_{b \in B \setminus T} b$. Let $S \subseteq A$ be the multiset containing all elements of $T$ and $|B| - |T|$ additional zeros. Clearly, $|S| = |B| = |A| / 2$ and $\sum_{a \in S} a = \sum_{b \in T} b = \sum_{b \in B} b / 2 = \sum_{a \in A} a / 2$, meaning that $A$ is a YES instance of {\normalfont \scshape Balanced 2-Partition} as desired.

(NO Case) We prove the contrapositive; suppose that $A$ is a YES instance of {\normalfont \scshape Balanced 2-Partition}. This means that there exists $S \subseteq A$ of size $|A| / 2 = |B|$ such that $\sum_{a \in S} a = \sum_{a \in B \setminus S} a$. Let $T$ be the subset of $B$ containing all elements of $S$ whose corresponding elements are included in $B$. Clearly, we have $\sum_{b \in T} b = \sum_{a \in S} a = \sum_{a \in B \setminus S} a = \sum_{b \in B \setminus T} b$. Hence, $B$ is a YES instance of {\normalfont \scshape 2-Partition}. 
\end{proof}

\end{document}